\documentclass[fleqn,11pt]{article}
\usepackage{fullpage}
\usepackage{a4wide}
\usepackage{amsmath,amssymb,amsfonts,amsthm}
\usepackage[fleqn,tbtags]{mathtools}
\usepackage[english]{babel}

\usepackage{graphicx}
\usepackage{epsfig}
\usepackage{color}
\usepackage{xspace}
\usepackage{url}
\usepackage{ifpdf}
\ifpdf
\usepackage{hyperref}
\usepackage[utf8]{inputenc} % allow utf-8 input
\usepackage{booktabs}       % professional-quality tables
\else
\usepackage[hypertex]{hyperref}
\fi
\usepackage [autostyle, english = american]{csquotes}
\usepackage[T1]{fontenc}
\MakeOuterQuote{"}

\newtheorem{theorem}{Theorem}%[section]
\newtheorem{lemma}[theorem]{Lemma}
\newtheorem{corollary}[theorem]{Corollary}
\newtheorem{proposition}[theorem]{Proposition}
\newtheorem{claim}[theorem]{Claim}

\newtheorem{assumption}{Assumption}
\makeatletter
\newtheorem*{rep@theorem}{\rep@title}
\newcommand{\newreptheorem}[2]{%
\newenvironment{rep#1}[1]{%
 \def\rep@title{#2 \ref{##1}}%
 \begin{rep@theorem}}%
 {\end{rep@theorem}}}
\makeatother

\newreptheorem{theorem}{Theorem}
\newreptheorem{lemma}{Lemma}

\newcommand {\ignore} [1] {}

\def \eqdef {:=}

\DeclareMathOperator{\poly}{poly}

\providecommand{\eqdef}{:=}

\newcommand{\etal}{{\em et al.\ }\xspace}

\newcommand{\ProblemName}[1]{\textsf{#1}}
\newcommand{\SCO}{\ProblemName{Set Cover}\xspace}
\newcommand{\KTR}{\ProblemName{$k$-Tree}\xspace}
\newcommand{\NTR}{\ProblemName{$n$-Tree}\xspace}
\newcommand{\KPT}{\ProblemName{$k$-Path}\xspace}
\newcommand{\KHP}{\ProblemName{$k$-HyperPath}\xspace}
\newcommand{\HAMP}{\ProblemName{Hamiltonian Path}\xspace}

\newcommand{\kMLD}{\ProblemName{$k$-MLD}\xspace}

\newcommand{\EXC}{\ProblemName{Exact Cover}\xspace}

\newcommand{\SPA}{\ProblemName{Set Partitioning}\xspace}
\newcommand{\KHC}{\ProblemName{kHyperCycle}\xspace}

\newcommand\tOm{\ensuremath{\tilde \Omega}}

\providecommand{\card}[1]{\lvert#1\rvert}

\title{Nearly Optimal Time Bounds for \texorpdfstring{$k$-Path}{kPath} in Hypergraphs}

\author{Lior Kamma \thanks{Computer Science Department. Aarhus University. Supported by a Villum Young Investigator Grant. \texttt{lior.kamma@cs.au.dk}.} \qquad 
Ohad Trabelsi\thanks{Weizmann Institute of Science, Rehovot, Israel \texttt{ohad.trabelsi@weizmann.ac.il}.}
}

\begin{document}
\maketitle

\begin{abstract}
We give almost tight conditional lower bounds on the running time of the \KHP problem.
Given an $r$-uniform hypergraph for some integer $r$, \KHP seeks a tight path of length $k$. That is, a sequence of $k$ nodes such that every consecutive $r$ of them constitute a hyperedge in the graph.
This problem is a natural generalization of the extensively-studied \KPT problem in graphs.
We show that solving \KHP in time $O^*(2^{(1-\gamma)k})$ where $\gamma>0$ is independent of $r$ is probably impossible. Specifically, it implies that \SCO on $n$ elements can be solved in time $O^*(2^{(1 - \delta)n})$ for some $\delta>0$. 
The only known lower bound for the \KPT problem is $2^{\Omega(k)}\poly(n)$ where $n$ is the number of nodes assuming the Exponential Time Hypothesis (ETH), and finding any conditional lower bound with an explicit constant in the exponent has been an important open problem.

We complement our lower bound with an almost tight upper bound. Formally, for every integer $r\geq 3$ we give algorithms that solve \KHP and \KHC on $r$-uniform hypergraphs with $n$ nodes and $m$ edges in time $2^k m \cdot\poly(n)$ and $2^k m^2 \poly(n)$ respectively, and that is even for the directed version of these problems. 
To the best of our knowledge, this is the first algorithm for \KHP.
The fastest algorithms known for \KPT run in time $2^k\poly(n)$ for directed graphs (Williams, 2009), and in time $1.66^k\poly(n)$ for undirected graphs (Bj\"orklund \etal, 2014).
\end{abstract}

\section{Introduction}
In the \KPT problem, given a graph $G$ and an integer $k$, the goal is to decide whether $G$ contains a simple path of length $k$. This is a fundamental combinatorial optimization problem, and contains the \HAMP as a special case. The fastest algorithm known for directed graphs runs in time $2^k\poly(n)$ \cite{Will09} and for the undirected version the fastest algorithm runs in time $1.66^k\poly(n)$ \cite{BHPK17}. The only known lower bound for these problems is $2^{\Omega(k)}\poly(n)$ assuming the Exponential Time Hypothesis (by a simple reduction from Hamiltonian Path), and finding any conditional lower bound with explicit constant in the exponent is an important open problem. 

Motivated by this challenge, we study the \KHP problem, which is a natural generalization of \KPT. Given an $r$-uniform hypergraph $G$, the goal is to find a tight path in $G$ of length $k$. That is, a sequence of $k$ nodes such that every $r$ consecutive nodes constitute an edge in the graph. 
The cycle variant of problem was previously studied by Lincoln \etal \cite{LWW18} for large values of the uniformity parameter $r=\Omega(k)$, where a conditional lower bound of $\tOm(n^k)$ was presented.
Considering smaller values of $r$, we show that for every $\gamma>0$ there is some integer $r \ge 3$ such that solving \KHP in time $O^*(2^{(1-\gamma)k})$ on $r$-uniform hypergraphs is most likely impossible. Formally,  we present a hierarchy of conditional lower bounds for \KHP in undirected hypergraphs with explicit constants in the exponents, where these lower bounds approach $O^*(2^k)$ as $r$ increases.

We complement our conditional lower bounds with almost matching upper bounds. That is, for every integer $r\geq 3$ we show an algorithm with running time $2^k m \cdot\poly(n)$ for \KHP on $r$-uniform hypergraphs with $n$ nodes and $m$ edges even for the directed version of this problem, where every hyperedge is a sequence of $r$ nodes, and tight paths respect the ordering of the corresponding edges. 

\subsection{Our Results}

Our main result shows that \KHP cannot be solved faster than $O^*(2^k)$ by an exponential factor independent of the size $r$ of a hyperedge unless \SCO on $n$ elements can be solved significantly faster than $O^*(2^n)$. The latter statement is known as the \SCO conjecture, introduced by Cygan \etal \cite{CDLMNOPSW16}. Formally, we show the following.
\begin{theorem} \label{th:lowerBound}
Let $\gamma>0$ and let $r\geq 3$ be an integer. If \KHP in undirected $r$-uniform hypergraphs can be solved in time $O^*(2^{(1-1/(r-1)-\gamma)k})$, then there exists $\delta=\delta(\gamma)>0$, such that \SCO on $n$ elements can be solved in time $O^*(2^{(1-\delta)n})$.
\end{theorem}

We prove Theorem~\ref{th:lowerBound} by showing a reduction from \SCO to \KHP. The reduction consists of two main parts. We first show that \SCO can be reduced to \EXC, thus implying that \EXC is "at least as hard" as \SCO.

\begin{lemma}\label{l:scoToExc}
Let $\gamma>0$. If \EXC on $n$ elements can be solved in time $O^*(2^{(1-\gamma)n})$, then there exists a constant $\delta>0$ such that \SCO on $n$ elements can be solved in time $O^*(2^{(1-\delta)n})$.
\end{lemma}
To the best of our knowledge, this "\SCO hardness" of \EXC was not known before. As \EXC seems to be a more robust source of reductions than its optimization variant \SPA, due to lack of budget constraints, this result may be of independent interest.

Secondly, we show that \EXC can be reduced to \KHP. Combining these two reduction constitutes the proof of Theorem~\ref{th:lowerBound}.
\begin{lemma}\label{l:excToKhp}
Let $\gamma>0$ and let $r\geq 3$ be an integer. If \KHP in undirected $r$-uniform hypergraphs can be solved in time $O^*(2^{(1-1/(r-1)-\gamma)k})$, then there exists a constant $\delta >0$, such that \EXC on $n$ can be solved in time $O^*(2^{(1-\delta)n})$.
\end{lemma}
Lemma~\ref{l:excToKhp} as well as its proof demonstrate that \EXC may be the right problem to associate with \KPT, given the connection it shows between its generalization \KHP and \EXC.
We note that by simple modifications to our proofs, all of our results can also be applied for the problem of finding cycles rather than paths with some small (polynomial in the input size) overhead. In addition, our conditional lower bounds for the undirected case imply the same bounds for the directed one by the following simple reduction. Given an $r$-uniform undirected hypergraph $H$, construct an $r$-uniform directed hypergraph $\vec{H}$ by including for every edge $e \in H$ all possible orientations of $e$. Clearly there is a $k$-path in $\vec{H}$ if and only if there is one in $H$. Moreover, the size of $\vec{H}$ is at most $r!$ times the size of $H$.

We accompany our lower bound by an almost tight upper bound showing that even directed \KHP is not too difficult a generalization of \KPT. 
While directed \KPT can be solved in time $O^*(2^k)$, it seems that there is no trivial way to extend this algorithm to \KHP. Our second result shows such an extension. That is, for every fixed integer $r$, directed $r$-uniform \KHP admits an algorithm with running time $2^k m n^{O(1)}$. 
\begin{theorem}\label{th:alg}
For every integer $r\geq 3$, directed $r$-uniform \KHP (respectively \KHC) can be solved in time $2^k m n^{O(1)}$ for some universal constant $c>0$ (respectively $2^k m^2 n^{O(1)}$), where $n$ and $m$ are the number of nodes and hyperedges respectively.
\end{theorem}
Our results show that while \KHP can essentially be solved as fast as \KPT, we know that for the former it is nearly tight, assuming the hardness of \SCO.

Note that once again by the previously described reduction our algorithm could be applied to the easier undirected case with the same running time (for $r=O(1)$). 
Our algorithm is achieved through a reduction to the $k$-Multilinear Monomial Detection problem (\kMLD), where the goal is to detect multilinear monomials of degree $k$ in a polynomial presented as a circuit. 
This problem has been utilized to solve \KPT~\cite{Koutis08,Will09}, and our method can be seen as an extension to \KHP.
Our method uses as a black box an algorithm for \kMLD with running time $2^k s(n) n^{O(1)}$ by Williams~\cite{Will09}, where $s(n)$ is the size of the circuit.

\subsection{Previous Work}
\KPT can be solved naively in time $\tilde{O}(n^k)$, however a long line of work devoted effort to find faster algorithms, starting with $O^*(f(k))$-time algorithms, specifically $O^*(k!)$~\cite{monien85} and $O^*(k!2^k)$~\cite{bod93}, followed by a series of improvements~\cite{kneis06,chen07} with the color coding method~\cite{Alon95} being a notable one. Finally, the fastest methods that were developed~\cite{Koutis08,Will09} utilize \kMLD to create \KPT algorithms with running times $O^*(2^{3k/2})$ and $O^*(2^k)$ respectively.

On the hardness front, it is well known that \KPT requires time $2^{\Omega(k)}\poly(n)$ assuming ETH,
as Hamiltonian path is a special case of this problem and 
there is a reduction from \ProblemName{$3$-SAT} to Hamiltonian Path with the number of nodes in the produced instance linear in the formula size.
If we care about the exact exponent in the running time, only a restricted lower bound is known.
Koutis and Williams~\cite{Koutis16} used communication complexity to show that a faster algorithm for their intermediate problem \kMLD in some settings is not possible, and so among a specific class of algorithms, their $O^*(2^k)$ algorithms for \KPT and \KTR, a generalization of \KPT whose goal is finding an isomorphic copy of a given tree of size $k$ in a given graph, is optimal. Krauthgamer and Trabelsi \cite{KT18} show that \KTR also requires $\Omega^*(2^{k})$ time assuming \SCO requires $\Omega^*(2^{n})$ time. In the other direction, they show that if \SCO on sets of size bounded by $O(\log n)$ can be solved in time significantly faster than $O^*(2^{n})$, then also \NTR can be solved in time significantly faster than $O^*(2^{n})$, where \NTR is \KTR but with a pattern tree that has number of nodes equals $\card{V(G)}$ (thus \NTR generalizes Hamiltonian Path).

\paragraph{The Set Cover Conjecture.} The set-cover conjecture formally states that for every fixed $\varepsilon>0$ there is an integer $\ell=\ell(\varepsilon)>0$ such that \SCO with sets of size at most $\ell$ cannot be solved in time $O^*(2^{(1-\varepsilon)n})$. The conjecture clearly implies that for every fixed $\varepsilon>0$, \SCO cannot be solved in time $O^*(2^{(1-\varepsilon)n})$.
In spite of extensive effort, the fastest algorithm for \SCO is still essentially a dynamic programming algorithm that runs in time $O^*(2^n)$~\cite{fomin04}, with several improvements in special cases~\cite{Koivisto09, Bjorklund09, N16, BHPK17}. Several conditional lower bounds were based on this conjecture in the recent decade, including for \ProblemName{Set Partitioning}, \ProblemName{Connected Vertex Cover}, \ProblemName{Steiner Tree}, \ProblemName{Subset Sum}~\cite{CDLMNOPSW16} (though the last problem was later shown hard conditioned on the strong exponential time hypothesis~\cite{ABHS19}), \ProblemName{Maximum Graph Motif}~\cite{bjor16}, parity of the number of solutions to \SCO with at most $\ell$ sets~\cite{BHH15}, \ProblemName{Colorful Path} and \ProblemName{Colorful Cycle}~\cite{kow16}, the dynamic, general and connected versions of \ProblemName{Dominating Set}~\cite{kri17}, and \KTR~\cite{KT18}.

\section{Preliminaries}
\paragraph{Paths in Hypergraphs.}
As there are several ways to generalize simple paths from ordinary graphs to hypergraphs we define it as follows, extending known definitions from the Hamiltonicity context. Let $H = (V, E)$ be an $r$-uniform
hypergraph. For every $1 \le \ell \le r-1$ and $k$ such that $(r-\ell)$ divides $(k-\ell)$, an $\ell$-overlapping $k$-path is a sequence of $k$ distinct nodes and $(k-\ell)/(r-\ell)$ hyperedges such that each hyperedge consists of $r$ consecutive nodes, and every pair of consecutive edges $e, e'$ satisfies $|e \cap e'| = \ell$. In the case where $\ell = r-1$, such paths are called {\em tight} paths, and
throughout we will focus on such paths. This natural generalization
of paths to hypergraphs has been extensively studied in the area of Hamiltonian paths and cycles,
either in the context of algorithms (e.g., \cite{FKL11,GM16}) or combinatorics (e.g., \cite{DFRS17,HZ15}). In particular,
Lincoln \etal \cite{LWW18} show that finding a tight cycle of size $k$ in $r$-uniform hypergraphs for $r= k -\left\lceil k/r' \right\rceil + 1$
cannot be solved significantly faster than $\tilde{O}(n^k)$ unless \textsf{MAX-$r$-SAT} can be solved significantly faster than $O^*(2^n)$, and also $r'$-uniform \textsf{$k$-HyperClique} can be solved significantly faster than $\tilde{O}(n^k)$.

\paragraph{\SCO, \EXC and \SPA.} For sake of completeness of the text, and to avoid confusion, we give a formal definition for each of the three set covering problems discussed in this paper. The input for all three problems is the same. We are given a ground set $U$ and a family ${\cal S} \subseteq 2^U$ of subsets of $U$. A {\em sub-cover} of $U$ from ${\cal S}$ is a subfamily ${\cal S}'\subseteq {\cal S}$ whose union is $U$. In the \SCO problem the goal is to find a sub-cover of minimal size. The \EXC problem is a decision problem that seeks to find whether there exists some sub-cover composed of pairwise disjoint subsets of $U$, also called a {\em partition}. Finally, the \SPA problem is the optimization variant of \EXC. That is, the goal is to find a partition of minimal size.

\section{A Conditional Lower Bound for \texorpdfstring{\KHP}{KHP}}\label{lowerBound}
This section is devoted to proving Theorem~\ref{th:lowerBound} by presenting a reduction from \SCO to \KHP. 
The reduction is presented in two steps. The first step, in which we prove Lemma~\ref{l:scoToExc}, reduces \SCO to \EXC. 
In the second step, which constitutes the technical crux of the proof, we present a reduction from \EXC to \KHP, thus proving Lemma~\ref{l:excToKhp}. 

\subsection{Reduction from \SCO to \EXC}\label{Section:SCO_EXC}
In this section we show that if \EXC can be solved in time significantly better than $O^*(2^n)$, then so does \SCO, thus proving Lemma~\ref{l:scoToExc}. More formally, we assume that there exist $c, \gamma > 0$ such that \EXC on $n$ elements and $m$ sets can be solved in time $O(m^c2^{(1-\gamma)n})$ and construct an algorithm that solves \SCO on $n$ elements and $m$ sets in time $O(\poly(m)2^{(1 - \delta)n})$ for some $\delta = \delta(c,\gamma) \in (0,1/2)$ to be determined later.

By the self-reducibility property of \SCO, it is enough to show that the decision version of \SCO can be solved in time $O(\poly(m)2^{(1 - \delta)n})$. That is, given a \SCO instance and a threshold $t \in \mathbb{N}$, the goal is to decide if there is a set cover of size at most $t$. We will show that for large values of $t$, this problem can be solved using an algorithm by Nederlof \cite{N16}. For small values of $t$, we use a reduction by Nederlof to the decision version of \SPA, and then reduce this problem to solving many (but not too many) instances of \EXC.
To this end, and following the notation suggested by Nederlof, we refer to instances of the decision variant of \SCO and \SPA as $(n,m,t)$-instances to denote that the number of elements in the ground set is $n$, the number of subsets given is $m$ and the threshold given is $t$.

Assume therefore that we are given an $(n,m,t)$-instance of \SCO. We first note that if $t \ge \sqrt[\leftroot{-2}\uproot{2}4]{\delta} n$, then the following result by Nederlof solves the problem in time $O(\poly(m)2^{(1 - \delta)n})$.

\begin{lemma}[\cite{N16}]\label{l:setCoverLarge}
There is a Monte-Carlo algorithm that takes a \SCO instance on $n$ elements and $m$ sets, as well as an integer $t$, and determines in time $O^*(2^{(1 - (t/n)^4)n})$ if there is a set cover of size $t$.
\end{lemma}

Assume therefore that $t \le \sqrt[\leftroot{-2}\uproot{2}4]{\delta} n$.
Applying the following reduction by Nederlof, we construct an $(n,m',t)$-instance of \SPA with $m'\le m2^{\delta n}$.
\begin{lemma}[\cite{N16}]
There is an algorithm that, given a real $0<\delta<1/2$, takes an $(n,m,t)$-instance of \SCO as input and outputs an equivalent $(n,m',t)$-instance of \SPA with $m'\leq m2^{\delta n}$ sets in time $O(m2^{(1-\delta)n})$.
\end{lemma}

Given the $(n,m',t)$-instance of \SPA constructed above we then construct $n 2^{2t}$ instances of \EXC by repeating the following color-coding scheme independently at random. Sample a random coloring $f : m' \to [t]$ of the input sets to $t$ colors. Add $t$ new elements, each one associated with a different color, and add each such element to all subsets of that color. Note that this \EXC instance has $n' = n + t \le (1 + \sqrt[\leftroot{-2}\uproot{2}4]{\delta}) n$ elements, and $m'\le m2^{\delta n}$ sets. The following lemma shows that this reduction reduces \SPA to \EXC.

\begin{lemma}
If the \SPA is a 'yes' instance, then with high probability at least one of the \EXC instances is a 'yes' instance. Conversely, if one of the \EXC instances is a 'yes' instance, then the \SPA instance is a 'yes' instance with certainty.
\end{lemma}

\begin{proof}
Assume first that the $(n,m',t)$-instance of \SPA is a 'yes' instance. Then there exists a partition of size $\ell \le t$. With probability at least $\frac{t!}{(t - \ell)!t^\ell} \ge \frac{t!}{t^t} \ge e^{-t}$, a uniform coloring of the sets colors each set of the partition with a different color. Therefore for every new element added, there is at most one set in the partition with the same color, and therefore there is an exact cover in the corresponding instance. Since the reduction constructs $n \cdot 2^{2t}$ independent instances, with high probability at least one will be a 'yes' instance.
Conversely, assume that one of the \EXC instances is a 'yes' instance, and consider an exact cover for this instance. Since the sets in the cover are disjoint, for every new element added, there is at most one subset in the cover containing it. We can therefore conclude that no two subsets are colored with the same color. Removing the new elements yields a feasible solution to the \SPA instance in which every subset is colored differently. Therefore the cover consists of at most $t$ subsets.
\end{proof}
By our original assumption, each of the \EXC instances can be solved in time $m'^c \cdot 2^{(1-\gamma)n'}$. Therefore the total time required to solve all instances is at most
$$n2^{2t} \cdot O\left(m'^c \cdot 2^{(1-\gamma)n'}\right) \le O\left(n2^{2\sqrt[\leftroot{-2}\uproot{2}4]{\delta} n} \cdot m^c2^{c \delta n} \cdot 2^{(1 - \gamma)(1 + \sqrt[\leftroot{-2}\uproot{2}4]{\delta}) n}\right) = O^*\left(2^{(2\sqrt[\leftroot{-2}\uproot{2}4]{\delta} + c \delta + (1 - \gamma)(1 + \sqrt[\leftroot{-2}\uproot{2}4]{\delta}))n} \right)\;.$$
For small enough choice of $\delta$ (depending only on $\gamma$ and $c$), this is at most $O\left(nm^c2^{(1- \delta)n}\right)$. This completes the proof of Lemma~\ref{l:scoToExc}.

\subsection{Reduction from \texorpdfstring{\EXC}{EXC} to \texorpdfstring{\KHP}{KHP}}\label{ProofsKCN}

In this section we show that if \KHP can be solved in time significantly better than $O^*(2^k)$, then so does \EXC, thus proving Lemma~\ref{l:excToKhp}. More formally, we show that if there exists a constant $\gamma > 0$ and an integer $r \ge 3$ such that \KHP on can be solved in time $O^*(2^{(1-1/(r-2)-\gamma)k})$ then \EXC on $n$ elements and $m$ sets can be solved in time $O(\poly(m)2^{(1 - \delta)n})$ for some $\delta = \delta(\gamma) \in (0,1/2)$ to be determined later.

To this end, let $X = \{x_1,\ldots,x_n\}$ and ${\cal S} = \{S_1,\ldots,S_m\}$ be an \EXC instance, and let $r \ge 3$ be fixed. We will present a procedure that runs in time polynomial in $m,n$ and constructs a \KHP instance on an undirected $r$-uniform hypergraph $H=(V_H,E_H)$ with $|V_H|=O(mn)$ nodes and $k = O(n)$.
We first present the construction under the following assumption. For sake of fluency we will show how to discard it after presenting the construction.
\begin{assumption} \label{ass:modr}
$n \ge 4r$, $r-2$ divides $n+2$, and for every $i \in [m]$, $|S_i| \ge 2r$.
\end{assumption}
We start by defining the vertex set $V_H$. For every $j \in [n]$, the hypergraph has a node labeled $x_j$ (the distinction between the node and the corresponding element will be clear from the context). In addition, for every $i \in [m]$ and $j \in [|S_i|]$ we define a node labeled $u_i^j$. The former set of nodes will be referred to as {\em element-nodes} and the latter one as {\em set-nodes}. Finally, we add two more element-nodes labeled $x_{start}, x_{end}$ and two set-nodes labeled $u_{start}, u_{end}$.
Next we turn to define the set $E_H$ of hyperedges, also demonstrated in Figure~\ref{Figures:reduction}. Loosely speaking $H$ has four types of edges. {\em Internal hyperedges} consist of nodes associated with a single set, {\em transition edges} consist of nodes associated with a two disjoint sets, {\em starting hyperedges} (resp. {\em ending hyperedges}) consist of nodes associated with a single set together with a starting (resp. ending) node.

\begin{figure}[!ht]
	\centering
		\includegraphics[width=1.0\textwidth]{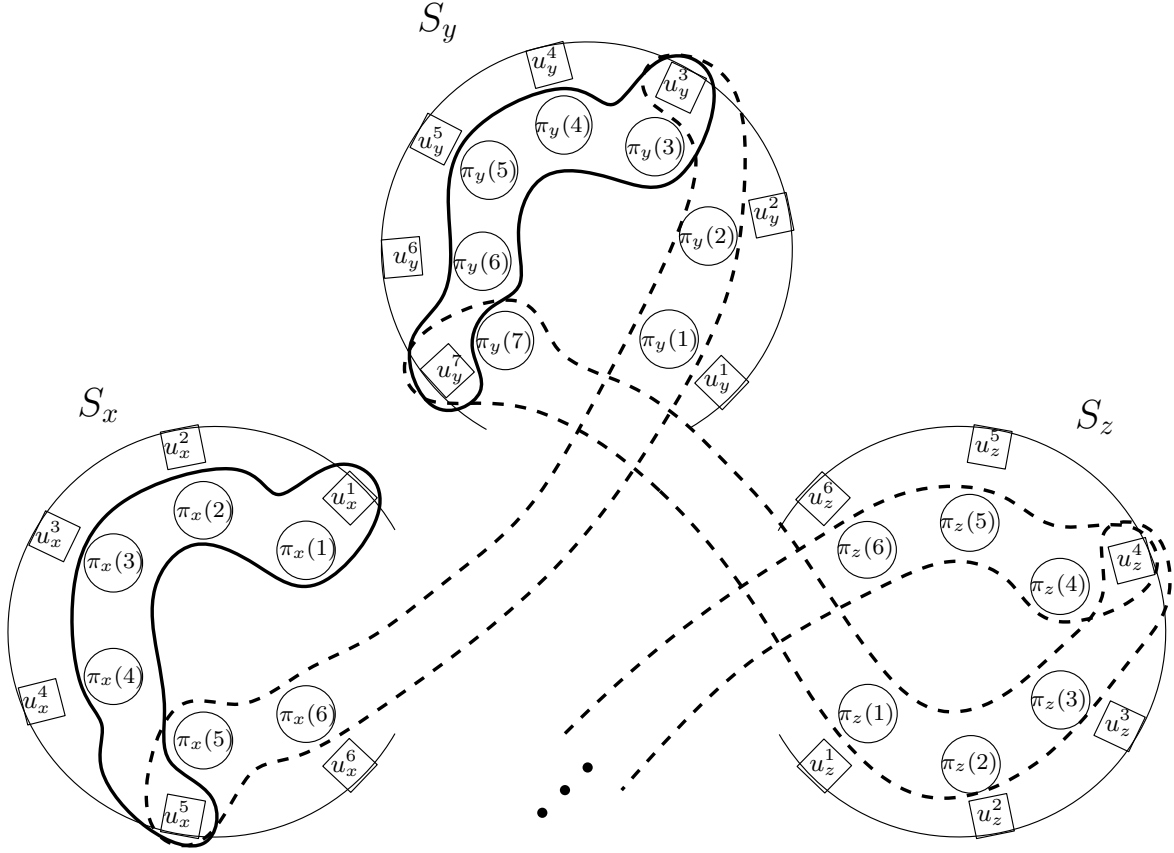}
   \caption[-]{An illustration of the reduction with $r=6$. The figure portrays three disjoint edges, as well as two heavy internal edges and two heavy transition edges. Round nodes represent element-nodes, which are associated with the ground set of elements $X$, and square nodes represent set-nodes, that are associated with the sets $S_1,S_2,S_3$. Continuous curves correspond to internal hyperedges, while dashed curves correspond to transition hyperedges.
   }
   \label{Figures:reduction}
\hrule
\end{figure}

To define the edges formally, for every $i \in [m]$ and $j \in [|S_i|]$ we let $x_i^j$ denote the $j$th element in $S_i$ \footnote{e.g. if $S_2 = \{x_3,x_5,x_{11}\}$ then $x_2^3=x_{11}$.}. We emphasize that $x_i^j$ is simply a new notation for an element in $X$. Therefore while for every $i \ne i'$ and $j \in [|S_{i}|], j' \in [|S_{i'}|]$, $u_i^j \ne u_{i'}^{j'}$, it might be the case that $x_i^j = x_{i'}^{j'}$.
Loosely speaking, the element-nodes are a "common resource" shared by all sets, while each set has "its own" set-nodes. Finally, although the edges are undirected, intuitively it might prove useful for the reader to think of the edge as a sequence, rather than a set.

Every internal edge contains a sequence of either $r-2$ or $r-1$ elements-nodes, all belonging to the same set. In the former case, the edge begins and ends with a set-node. In the latter, the edge contains a set-node in the midst of the sequence. Formally, for every $i \in [m]$ we define internal edges as follows.
\begin{align*}
&\{u_i^j\} \cup \{x_i^\ell\}_{\ell = j}^{j+r-3} \cup\{u_i^{j+r-2}\} \quad \quad & \quad \quad \forall  j \in [|S_i| - (r-2)] \\
&\{x_i^\ell\}_{\ell=j}^{j+h} \cup \{u_i^{j+h+1}\} \cup \{x_i^\ell\}_{\ell=j+h+1}^{j+r-2} \quad \quad & \quad \quad \forall j \in [|S_i| - (r-2)], h \in [0,r-3]
\end{align*}

Let $1 \le i<i' \le m$ be such that $S_i,S_{i'}$ are disjoint. A transition edge contains a sequence of either $r-2$ or $r-1$ element-nodes, the first part of which consists of elements of $S_i$ and ends with $x_i^{|S_i|}$, and the second is either empty or consists of elements of $S_{i'}$ starting with $x_{i'}^1$. In the former case, the edge begins with a set-node associated with $S_i$ and ends with a set-node associated with $S_{i'}$. In the latter, the edge either contains a set-node associated with $S_i$ in the midst of the subsequence associated with $S_i$ elements or a set-node associated with $S_{i'}$ in the midst of the subsequence associated with $S_{i'}$ elements. Formally, for all $i < i'$ such that $S_i \cap S_{i'}= \emptyset$ the transition hyperedges are defined as follows.
\begin{align*}
&\{u_i^{|S_i|-j}\} \cup \{x_i^\ell\}_{\ell = |S_i|-j}^{|S_i|} \cup \{x_{i'}^\ell\}_{\ell = 1}^{r-3-j} \cup\{u_{i'}^{r-2-j}\}  &  \forall j \in [0, r-3] \\
&\{x_i^\ell\}_{\ell=|S_i| - j}^{|S_i| - j + h} \cup \{u_i^{|S_i|-j+h+1}\} \cup \{x_i^\ell\}_{\ell=|S_i|-j+h+1}^{|S_i|} \cup\{x_{i'}^{\ell}\}_{\ell=1}^{r-j-2}   & \forall j \in [r-3], h \in [0,j-1] \\
&\{x_i^\ell\}_{\ell=|S_i| - j+1}^{|S_i|} \cup \{x_{i'}^\ell\}_{\ell=1}^{h}\cup \{u_{i'}^{h+1}\} \cup\{x_{i'}^{\ell}\}_{\ell=h+1}^{r-j-1}  & \forall j \in [r-3], h \in [r-j-2]
\end{align*}
Finally, we define the starting and ending edges. Intuitively, these edges are destined to be the start and ending of the desired path. Formally, for every $i \in [m]$ we define the following.
\begin{align*}
&\{u_{start},x_{start}\} \cup \{x_i^\ell\}_{\ell=1}^{r-3} \cup \{u_i^{r-2}\}\\
&\{x_{start}\} \cup \{x_i^\ell\}_{\ell=1}^{r-2} \cup \{u_i^{r-2}\}\\
&\{u_i^{|S_i|-(r-3)}\} \cup \{x_i^\ell\}_{\ell=|S_i|-(r-4)}^{|S_i|} \cup \{u_{end},x_{end}\}\\
&\{u_i^{|S_i|-(r-3)}\} \cup \{x_i^\ell\}_{\ell=|S_i|-(r-3)}^{|S_i|} \cup \{x_{end}\}
\end{align*}
To complete the \KHP instance we set $k=(n+2)(1+1/(r-2))+1$. Note that by Assumption~\ref{ass:modr}, $k \ge 3r$ is a positive integer. The following claim is straightforward.
\begin{claim}\label{c:polyTime}
Given the \EXC instance $X,{\cal S}$, we can construct $H,k$ in polynomial time.
\end{claim}

We will now show that $H$ contains a path of length $k$ if and only if the \EXC instance is a 'yes' instance. We first show that if $H$ contains a path of length $k$ then the \EXC instance is a 'yes' instance. To this end, assume there is a tight path $$P = v_1-v_2-\ldots-v_k$$ of length $k$ in $H$. Let ${\cal S}_P \subset {\cal S}$ be the collection of all sets in ${\cal S}$ that have a set-node in $P$. Formally, let ${\cal I}_P \eqdef \{i \in [m] : \exists j \in [|S_i|]. \; u_i^j \in P\}$ then ${\cal S}_P \eqdef \{S_i : i \in {\cal I}_P\}$. We will show the following.
\begin{lemma}\label{l:exactCoverFromPath}
${\cal S}_P$ is an exact cover of $X$.
\end{lemma}

To prove the lemma we start with characterizing the structure of $P$. In what follows, we refer to hyperedges containing exactly two set-nodes as {\em heavy} and to hyperedges containing exactly one set-node as {\em light}. The following straightforward claim follows directly from the construction.

\begin{claim} \label{c:edgeIntersection}
Let $e,e' \in E_H$ be two hyperedges, and let $i \in [m]$ and $j \in [|S_i|]$. 
\begin{enumerate}
	\item $e \cap e'$ contains at most one set-node.
	\item If $e$ is light and $u_i^j \in e$ then $x_i^j \in e$. If $j \ge 2$ then $x_i^{j-1} \in e$.
	\item If $e,e'$ are both light, and $u_i^j \in e \cap e'$, then $e \cap e'$ contains at least two element-nodes.
	\item If $e,e'$ are both heavy, and $u_i^j \in e \cap e'$, then $e \cap e' = \{u_i^j\}$.
	%\item If $e$ is not a base edge, and $u_i^j$ is its (unique) set node, then $x_i^j \in e$.
\end{enumerate}
\end{claim}
\begin{lemma}
Let $1 \le \alpha < \beta \le k$ be such that $v_\alpha, v_\beta$ are consecutive set nodes in $P$. Then $\beta-\alpha = r-1$.
\end{lemma}
\begin{proof}
We will start by proving $\beta - \alpha \le r-1$. Since $P$ is tight, $e= \{v_\ell\}_{\ell=\alpha+1}^{\alpha+r} \in E_H$, and since every edge contains at least one set-node, $\beta \le \alpha + r$ thus $\beta - \alpha \le r$. Assume towards contradiction that $\beta - \alpha = r$. Since $k \ge 3r$, either $k - \beta \ge r-1$ or $\alpha \ge r$. Without loss of generality, assume $k - \beta \ge r-1$ (the other case is analogous). Denote $e'= \{v_\ell\}_{\ell=\beta-1}^{\beta+r-2}, e''=\{v_\ell\}_{\ell=\beta}^{\beta+r-1} \in E_H$. By Claim~\ref{c:edgeIntersection}, since $e', e''$ share the set-node $v_{\beta}$, then $v_{\beta+1}, v_{\beta+2}, \ldots, v_{\beta+ r-2}$ are all element-nodes. Therefore $e,e'$ both contain only one set-node, namely $v_{\beta}$, which they also share, however their intersection contains only one element-node, in contradiction to Claim~\ref{c:edgeIntersection}. Therefore $\beta-\alpha \le r-1$.

Next we prove that $\beta-\alpha \ge r-1$. Assume first that $\alpha > 1$. Then $\{v_\ell\}_{\ell=\alpha-1}^{\alpha+r-2}$ and $\{v_\ell\}_{\ell=\alpha}^{\alpha+r-1}$ are both edges of $H$ that share $v_\alpha$, and by Claim~\ref{c:edgeIntersection} cannot share $v_{\beta}$. Therefore $\beta \ge \alpha + r-1$ thus $\beta - \alpha \ge r-1$. Otherwise, if $\alpha=1$, then $k-\beta \ge 2r$. Since $\{v_\ell\}_{\ell=\beta+1}^{\beta+r} \in E_H$, there exists $\beta+1 \le \gamma \le \beta+r$ such that $v_{\gamma}$ is a set-node. By the first part of the proof, $\gamma-\beta \le r-1$, and since $\beta>1$, then by the proof for the case $\alpha>1$ we get that $\gamma - \beta = r-1$. Therefore $\{v_\ell\}_{\ell=\alpha}^{r}, \{v_\ell\}_{\ell=\beta}^{\beta+ r - 1}$ are both heavy, and they share $v_{\beta}$. By Claim~\ref{c:edgeIntersection}, they do not have any element-node in common, and therefore $r \le \beta$ thus $r-1 \le \beta-1 = \beta - \alpha$. We conclude that $\beta - \alpha = r$.
\end{proof}
We can now give a complete characterization of $P$ in terms of set-nodes and element-nodes.

\begin{lemma}\label{l:setElementNodesPath}
For every $\alpha \in [k]$, $v_\alpha$ is a set node if and only if $r-1$ divides $\alpha-1$.
\end{lemma}
\begin{proof}
Let $\alpha_1, \alpha_2,\ldots, \alpha_t$ be such that $v_{\alpha_1}, v_{\alpha_2}, \ldots, v_{\alpha_t}$ are the only set-nodes in $P$ in consecutive order. From the previous lemma it follows that between every two set-nodes there are exactly $r-2$ element-nodes. It follows in addition by induction that $\alpha_t-\alpha_1 = (t-1)(r-1)$. Therefore the number of element-nodes between $v_{\alpha_1}$ and $v_{\alpha_t}$ is $(t-1)(r-2)$. In addition, the first $\alpha_t-1$ nodes, as well as the last $k - \alpha_t$ are element-nodes. Therefore 
\begin{equation}
(t-1)(r-2) + (\alpha_1-1)+(k - \alpha_t) \le n+2
\label{eq:elementSmallerTotal}
\end{equation}
By the definition of $k$ we have that $$(k-\alpha_t)+(\alpha_t-j_1)+(\alpha_1-1) = k - 1 = \frac{n+2}{r-2}(r-1) \;,$$ and therefore 
\begin{equation}
n+2 = (r-2)\frac{k-1}{r-1} = (t-1)(r-2) + (r-2)\frac{(k-\alpha_t) + (\alpha_1-1)}{r-1}\;.
\label{eq:totalEqual}
\end{equation}
Combining \eqref{eq:elementSmallerTotal} and \eqref{eq:totalEqual} we get that $(\alpha_1-1)+(k - \alpha_t) \le \frac{r-2}{r-1}[(k-\alpha_t) + (\alpha_1-1)]$. It follows that $(\alpha_1-1)+(k - \alpha_t) = 0$, and therefore $\alpha_1=1$ and $\alpha_t=k$. The lemma now follows by simple induction.
\end{proof}

\begin{corollary} \label{cor:startEnd}
Every element-node appears exactly once in $P$. Moreover, $v_1=u_{start}$, $v_k=u_{end}$, $v_2 = x_{start}$ and $v_{k-1}=x_{end}$.
\end{corollary}
\begin{proof}
From Lemma~\ref{l:setElementNodesPath} the number of element-nodes in $P$ is $(r-2) \cdot (k-1)/(r-1) = n+2$. Since $P$ is simple, every element-node appears exactly once in $P$, and specifically $x_{start}$ and $x_{end}$. Therefore $P$ must contain at least one starting edge and one ending edge. For every $3 \le \alpha \le k-2$, if $v_\alpha$ is an element-node, then there is a set node $u$ such that there are at least $2$ edges on $P$ containing both $v_\alpha$ and $u$ and no other set-node. Since for every $v \in \{x_{start},x_{end}\}$ and for every set node $u$, there is exactly one edge in $H$ that contains both $v$ and $u$ and no other set-node, it follows that $\{x_{start}, x_{end}\} = \{v_2,v_{k-1}\}$. Without loss of generality, $x_{start} = v_2$ and $x_{end} = v_{k-1}$. Since every base hyperedge containing $x_{start}$ (resp. $x_{end}$) must contain $u_{start}$ (resp. $u_{end}$), it follows that $v_1 = u_{start}$ and $v_k = u_{end}$.
\end{proof}

%%%%%%%%%%%%%%%
In what follows, we show that all element nodes of a given set appear consecutively along the path, and moreover, every two sets whose nodes are visited by the path are, in fact,  disjoint.

\begin{lemma}\label{l:elementNodeFollowingSetNode}
Let $i \in {\cal I}_P, j \in [m]$ and let $\alpha \in [k]$ be such that $v_\alpha=u_i^j$. Then $v_{\alpha+1} = x_i^j$ and if $j \ge 2$ then $v_{\alpha-1}=x_i^{j-1}$.
\end{lemma}
\begin{proof} 
First note that since $v_\alpha=u_i^j$ then $r \le \alpha \le k-r+1$. 
We prove the claim by induction on $\alpha$.
If $\alpha=r$ then the edge $e=\{v_\ell\}_{\ell=1}^r$ is a starting edge, and therefore $j=r-2$. By Claim~\ref{c:edgeIntersection} every light edge that contains $u_i^{r-2}$ must also contain $x_i^{r-2}$ and $x_i^{r-3}$.
The only two nodes contained in all light edges in $P$ that contain $v_\alpha$ are $v_{\alpha-1}$ and $v_{\alpha+1}$, it follows that $\{x_i^j,x_i^{j-1}\}=\{v_{\alpha-1},v_{\alpha+1}\}$. Since $e$ is a heavy starting edge, $x_i^{r-2} \notin e$. Therefore $v_{\alpha-1}=x_i^{r-3}$ and $v_{\alpha+1}=x_i^{r-2}$.

Assume therefore that the claim holds for all $r \le \alpha' < \alpha$, and let $\beta = \alpha - (r-1) \ge r$. Then $\beta$ is also a set node, and by the induction hypothesis, the claim holds for $v_\beta$. Let $i' \in [m]$ and $j' \in [|S_{i'}|]$ be such that $v_\beta = u_{i'}^{j'}$.
If $i = i'$ then $e = \{v_\ell\}_{\ell=\beta}^\alpha$ is a heavy internal edge containing $u_i^{j'}, u_i^j$. By the induction hypothesis, $x_i^{j'+1} = v_{\beta+1} \in e$. Therefore by the construction of internal edges, $j = j' + r -2 \ge 2$, and $x_i^{j-1} \in e$ and $x_i^{j} \notin e$. By arguments similar to the first part of the proof we get that $\{x_i^j,x_i^{j-1}\}=\{v_{\alpha-1},v_{\alpha+1}\}$, and therefore $v_{\alpha-1}=x_i^{j-1}$ and $v_{\alpha+1}=x_i^j$.

Otherwise, if $i \ne i'$, then $e = \{v_\ell\}_{\ell=\beta}^\alpha$ is a heavy transition edge containing $u_{i'}^{j'}, u_i^j$. Again, by the induction hypothesis, $v_{\beta+1} = x_{i'}^{j'+1}$. By the construction of transition hyperedges we get that $i' < i$ and $S_{i'} \cap S_i = \emptyset$. If $j \ge 2$ then in addition $x_i^{j-1} \in e$, and by arguments similar to the first part of the proof $v_{\alpha-1}=x_i^{j-1}$ and $v_{\alpha+1}=x_i^j$. If $j=1$ then $e'=\{v_\ell\}_{\ell=\beta+1}^{\alpha+1}$ is a transition light edge such that $x_i^1 \in e'$. Since $x_i^1 \notin e$ we get that $v_{\alpha+1}=x_i^1$. 
\end{proof}
The following follows from the previous lemma by induction.
\begin{proposition} \label{p:elementNodesOrder}
Let $i \in {\cal I}_P, j \in [|S_i|]$ and let $\alpha \in [k]$ be such that $u_i^j = v_\alpha$. Then 
\begin{enumerate}
	\item For every $1 \le \ell \le \min\{r-2, j-1\}$, $v_{\alpha-\ell} = x_i^{j-\ell}$ ; and
	\item For every $1 \le \ell \le \min\{r-2, |S_i|-j+1\}$, $v_{\alpha+\ell} = x_i^{j+\ell-1}$.
\end{enumerate}
Moreover, if $j-1 \ge r-2$ then $v_{\alpha-r+1}=u_i^{j-r+1}$ and if $|S_i|-j+1 \ge r-2$ then $v_{\alpha+r-1}=u_i^{j+r-2}$.
\end{proposition}

\begin{corollary} \label{cor:elementNodesConsecutive}
Let $i \in {\cal I}_P$, then for every $x,x' \in S_i$ the only nodes between $x$ and $x'$ on $P$ are either element-nodes or set-nodes that are associated with $S_i$.
\end{corollary}
We are now ready to prove Lemma~\ref{l:exactCoverFromPath}.
\begin{proof}[Proof of Lemma~\ref{l:exactCoverFromPath}]
We will first show that $\bigcup_{i \in {\cal I}_P}S_i = X$. Let $x \in X$. Then there is some $\alpha \in [k]$ such that $x = v_\alpha$. Let $\beta < \alpha < \gamma$ be such that $v_\beta$ and $v_\gamma$ are the set-nodes preceding and following $v_\alpha$ in $P$ respectively. Then $x \in \{v_\ell\}_{\ell=\beta}^\gamma$. By the construction of $H$, there exist $i \in [m]$ and $j, j' \in [|S_i|]$ such that either $v_\beta$ or $v_\gamma$ is $u_i^j$ and $x = x_i^{j'}$. Since $u_i^j \in P$ then $i \in {\cal I}_P$ and therefore $x \in \bigcup_{i \in {\cal I}}S_i$.

Next, let $i,i' \in {\cal I}_P$ be such that $i \ne i'$. Then there exist $j \in [|S_i|]$ and $j' \in [|S_{i'}|]$ such that $u_i^j, u_{i'}^{j'} \in P$. That is, there exist $\alpha, \alpha' \in [k]$ such that $v_\alpha = u_i^j$ and $v_{\alpha'}= u_{i'}^{j'}$. Without loss of generality, we may assume $\alpha < \alpha'$, and there are no set-nodes associated with $S_i$ between $v_\alpha$ and $v_{\alpha'}$ (otherwise we take the last set-node associated with $S_i$ before $v_{\alpha'}$). Then $e=\{v_\ell\}_{\ell=\alpha}^{\alpha+r}$ is a heavy transition hyperedge. If $\alpha'=\alpha+r$ then by definition of transition hyperedges, $S_i \cap S_{i'} = \emptyset$. Otherwise, none of the element nodes in $e$ are associated with $S_{i'}$, and moreover, $x_i^{|S_i|} \in e$. From Corollary~\ref{cor:elementNodesConsecutive} it follows that the elements of $S_i$ are visited by $P$ in the order $x_i^1,x_i^2,\ldots,x_i^{|S_i|}$ and similarly the elements of $S_{i'}$ are visited in the order $x_{i'}^1,x_{i'}^2,\ldots,x_{i'}^{|S_{i'}|}$. Since $x_{i'}^1$ is visited by $P$ after $x_i^{|S_i|}$ then we get that $S_i \cap S_{i'} = \emptyset$.
\end{proof}

We have therefore proved that if there is a tight path of length $k$ in $H$, then the \EXC is a 'yes' instance.
Conversely, assume that the \EXC instance is a 'yes' instance and let ${\cal S}_C=\{S_1,...,S_t\}$ be an exact cover. Since ${\cal S}_C$ is an exact cover, for every $x \in X$ there exist a unique pair $i \in [t], j \in [|S_i|]$ such that $x = x_i^j$. We can therefore order the elements of $X$ according to a lexicographic ordering. That is for all $i,i' \in [s], j \in [|S_i|]$ and $j' \in [|S_{i'}|]$, $x_i^j$ comes before $x_{i'}^{j'}$ if and only if $i < i'$ or $i=i'$ and $j < j'$. Consider the following sequence of element-nodes.
$$x_{start}-x_1^1-\ldots-x_1^{|S_1|}-x_2^1-\ldots-x_2^{|S_2|}-x_3^1-\ldots-x_t^{|S_t|}-x_{end} \;.$$
Between every $r-2$ element-nodes, insert the corresponding set-node. That is, if the element-node that follows the set-node is $x_i^j$, then the set-node is $u_i^j$. Add $u_{start}, u_{end}$ in the beginning and end of the sequence respectively. The following claim is proved similarly to the claims used to prove the previous direction.
\begin{claim}
The sequence constructed above constitutes a tight path of length $k$ in $H$.
\end{claim}

We have therefore proved that there exists an exact cover if and only if there is a tight path of length $k$ in $H$, thus completing the proof of Theorem~\ref{th:lowerBound}.
\subparagraph*{Discarding Assumption~\ref{ass:modr}.} 
Given an \EXC instance $X = \{x_1,\ldots,x_n\}, {\cal S}=\{S_1,\ldots,S_m\}$, we can augment $X$ by $k$ new elements $x_{n+1},\ldots,x_{n+k}$ for some $4r \le k \le 5r$ such that $r-2$ divides $n+k+2$. Call the new set $X'$. We then create $2r$ collections of subsets ${\cal S}_1,\ldots,{\cal S}_{2r}$ by defining ${\cal S}_\ell \eqdef {\cal S} \cup \{\{x_{n+1}\},\ldots,\{x_{n+\ell}\},\{x_{n+\ell+1},\ldots,x_{n+k}\}\}$ for every $\ell \in [2r]$. Note that $m \le |{\cal S}_\ell| \le m+2r$. Next, for every $\ell \in [2r]$ we define a collection ${\cal S}'_\ell$ as follows. For every $2r$ pairwise-disjoint sets in ${\cal S}_\ell$, we add their union to ${\cal S}'_\ell$. The following is straightforward.
\begin{claim}
Given $X,{\cal S}$ and $\ell \in [2r]$ 
\begin{enumerate}
	\item we can construct $X',{\cal S}'_\ell$ in time $O(m^{2r+1})$;
	\item for every $\ell \in [2r]$ and every set $S \in {\cal S}'_\ell$ we get that $|S|\ge 2r$; and
	\item $X',{\cal S}'_\ell$ is a 'yes' instance for \EXC if and only if $X, {\cal S}$ is a 'yes' instance for \EXC and has an exact cover of size $2rt-\ell$ for some positive integer $t$.
\end{enumerate}
\end{claim}
Therefore $X,{\cal S}$ is a 'yes' instance for \EXC if and only if there exists $\ell \in [2r]$ such that $X',{\cal S}'_\ell$ is a 'yes' instance for \EXC. Performing the reduction to \KHP on every instance separately we get a reduction of a general \EXC instance to $2r$ \KHP instances.

%%%%%%%%%%%%%%%%%%%%%%%%%%%%%%%%%%%%%%%%%%%%%%%%%%%%%%%%%%%%%%%%%%%%
%%%%%%%%%%%%%%%%%%%%%%%%%%%%%%%%%%%%%%%%%%%%%%%%%%%%%%%%%%%%%%%%%%%%

\section{Algorithm for \texorpdfstring{\KHP}{KHP}}\label{ProofsEXC}
In this section we show that \KHP can be solved in time $2^k m n^{O(1)}$, thus proving Theorem~\ref{th:alg}. We follow the ideas of Koutis~\cite{Koutis08} and of Williams~\cite{Will09}, with additional ideas to take care of our generalization. A short introduction to arithmetic circuits follows.
\paragraph*{Arithmetic Circuits}
An arithmetic circuit over a specified ring $K$ is a directed acyclic graph with nodes labeled from $\{+,\times\}\cup\{x_1,\dots,x_n\}\cup K$, where $\{x_1,\dots,x_n\}$ are the input variables of the circuit. Nodes with zero out-degree are called \textit{output nodes}, and nodes with labels from $\{x_1,\dots,x_n\}\cup K$ are called \textit{input nodes}. The size of the circuit is the number of nodes in it. Clearly, every output node can be associated with a polynomial on $K[x_1,\dots,x_n]$.
A polynomial $p\in K[x_1,\dots,x_n]$ is said to contain a multilinear term if it contains a term of the form $c\prod_{i\in S}x_i$ for a non-empty $S\subseteq [n]$ in the standard monomial expansion of $p$. See~\cite{circuits97} for more details on arithmetic circuits.
\paragraph*{The Algorithm}
Given a hypergraph $H$, we define a polynomial $P_k(X)$ on the set of variables $X=\{x_u:u\in V(H)\}$.
$$
P_k(x_1,\dots,x_n)=\sum_{i_1,...,i_k\text{ is a tight walk in }H} x_{i_1}\cdot\cdot\cdot x_{i_k},
$$
where a \textit{tight walk} is a sequence of nodes (not necessarily distinct) such that every $r$ consecutive nodes are contained in some edge in $H$.
Clearly, there is a $k$-hyperpath iff $P_k(x_1,\dots,x_n)$ contains a multilinear term.
We first show that $P_k$ can be succinctly implemented by an arithmetic circuit.
\begin{lemma}
Let $H=(V_H,E_H)$ be a hypergraph on $n$ nodes and $m$ hyperedges, and let $k\leq n$ be an integer. Then $P_k(x_1,\dots,x_n)$ is computable by an arithmetic circuit of size $O(mk)$.
\end{lemma}

For every hyperedge $(u_{i_1},u_{i_2},\dots,u_{i_r})$ and an integer $t\in [n]$ we define a polynomial $f_H(u_{i_1},u_{i_2},\dots,$ $u_{i_r},t)$.
Intuitively, $f_H(u_{i_1},u_{i_2},\dots,u_{i_r},t)$ contains a sum of monomials where each one corresponds to all tight walks of length $t$ that end with $u_{i_1},u_{i_2},\dots,u_{i_r}$ and on the same set of nodes.
Formally, $f_H(u_{i_1},u_{i_2},\dots,u_{i_r},t)$ can be defined inductively as follows. At the base level of the induction, for every edge $(u_{i_1},u_{i_2},\dots,u_{i_r})\in E_H$ we initialize $f_H(u_{i_1},u_{i_2},\dots,u_{i_r},t)= x_{i_1}\cdots x_{i_r}$ for $t=r$, and $f_H(u_{i_1},u_{i_2},\dots,u_{i_r},t)=0$ for $t < r$.
Now, $f_H(u_{i_1},u_{i_2},\dots,u_{i_r},t+1)$ can be constructed by knowing $f_H(v_{i_1},v_{i_2},\dots,v_{i_r},t)$ for all edges $(v_{i_1},v_{i_2},\dots,v_{i_r})$, as follows.
$$
f_H(u_{i_1},u_{i_2},\dots,u_{i_r},t+1)=\sum_{u_{i_0}:(u_{i_0},u_{i_1},\dots,u_{i_{r-1}})\in E_H} f_H(u_{i_0},u_{i_1},\dots,u_{i_{r-1}},t)\cdot x_{i_{r}}
$$
Finally, $\tilde{P}_k$ which is the required polynomial, sums over all hyperedges $e$ the length-$k$ tight walks that end with $e$.
$$
\tilde{P}_k(x_1,\dots,x_n)=\sum_{(u_{i_1},u_{i_2},\dots,u_{i_{r}})\in E_H} f_H(u_{i_1},u_{i_2},\dots,u_{i_{r}},k)
$$

The size of the $\tilde{P}_k$ is bounded by $1$ plus the number of possible labellings of the form $(u_{i_1},u_{i_2},\dots,u_{i_r},t)$ as an input to $f_H$, which is bounded by $O(mk)$ since $f_H$ is only defined where the first $r$ entries are an edge in $E_H$. 
We use the following result by~\cite{Will09} to achieve our bound.
\begin{theorem}[see Theorem $3.1$ in~\cite{Will09}]
Let $P(x_1,\dots,x_n)$ be a polynomial of degree at most $k$, represented by an arithmetic circuit of size $s(n)$ with $+$ gates (of unbounded fan-in), $\times$ gates (of fan-in two) and no scalar multiplications. There is a randomized algorithm that on every $P$ runs in $2^k s(n) n^{O(1)}$ time, outputs yes with high probability if there is a multilinear term in the sum-product expansion of $P$, and always outputs no if there there is no multilinear term.
\end{theorem}
Since the degree of $\tilde{P}_k$ is clearly $k$, the time to construct the circuit is $2^k m n^{O(1)}$.
By a simple inductive argument it can be proved that $f_H(u_{i_1},\dots,u_{i_r},t)$ indeed contains a multilinear monomial of degree $t$ iff there exists a length-$t$ tight path that ends with $(u_{i_1},\dots,u_{i_r})$.

If we care about cycles, the only difference is that we also need to remember from what edge we have started the path, thus adding to $f_H$ an additional amount of $r$ entries, and the circuit size would increase to $O(m^2 k)$, and the running time to $2^k m^2 n^{O(1)}$, thus proving Theorem~\ref{th:alg}.

\section{Conclusion and Open Problems}\label{Section:Conclusion}
We have proved that \KHP could be solved in time $2^km n^{O(1)}$ and hard to solve significantly faster even for bounded $r$'s.  
It remains open to show for undirected \KHP an algorithm with running time $O^*(2^{(1-\gamma(r))k})$ (i.e. when $\gamma$ is a function of $r$). It is unclear how to extend the current fastest algorithms for undirected Hamiltonian path or \KPT to hypergraphs, and it may even be possible that a stronger $O^*(2^k)$ conditional lower bound holds for \KHP even for $r=3$.
Also, an important open problem is to prove any conditional lower bound with explicit constant in the exponent to \KPT. This task has been hard to complete, maybe because a (graphic) path does not seem to have enough information to make manipulations needed to solve \SCO. While \KHP is a stronger variant, it still appears close to \KPT and shows a possible direction to prove conditional lower bounds to this problem.

\newcommand{\etalchar}[1]{$^{#1}$}

\end{document}